\documentclass[journal]{IEEEtran}
\usepackage{dsfont,graphicx,color,epsf,epsfig,verbatim,amssymb,amsmath,amsthm}
\usepackage{latexsym,amsfonts,caption,array,cite,subfigure,multicol,multirow}
\newtheorem{lemma}{Lemma}[section]

\newtheorem{theorem}{Theorem}[section]

\begin{document}
\title{Nested Lattice Codes for Arbitrary Continuous Sources and Channels}
\author{\IEEEauthorblockN{Aria G. Sahebi and S. Sandeep Pradhan\\ \thanks{This work was supported by NSF grants CCF-0915619 and CCF-1116021.}}
\IEEEauthorblockA{Department of Electrical Engineering and Computer Science,\\
University of Michigan, Ann Arbor, MI 48109, USA.\\
Email: \tt\small ariaghs@umich.edu, pradhanv@umich.edu}}

\markboth{Draft}
{Sahebi \MakeLowercase{\textit{et al.}}: Nested Lattice Codes for Arbitrary Continuous Sources and Channels}
\maketitle

\begin{abstract}
In this paper, we show that nested lattice codes achieve the capacity of arbitrary channels with or without non-casual state information at the transmitter. We also show that nested lattice codes are optimal for source coding with or without non-causal side information at the receiver for arbitrary continuous sources.
\end{abstract}


\IEEEpeerreviewmaketitle

\section{Introduction}
\IEEEPARstart{L}{attice} codes for continuous sources and channels are the analogue of linear codes for discrete sources and channels and play an important role in information theory and communications. Linear/lattice and nested linear/lattice codes have been used in many communication settings to improve upon the existing random coding bounds \cite{korner_marton, nazer_gastpar, debuda_optimal_codes_structure, Linder_debuda_theorem, phiosof_zamir, dinesh_dsc, Sridharan_Jafarian_Lattice_Interference, Loeliger_Averaging_Bound}.\\
In \cite{debuda_optimal_codes_structure} and \cite{Linder_debuda_theorem} the existence of lattice codes satisfying Shannon's bound has been shown. These results have been generalized and the close relation between linear and lattice codes has been pointed out in \cite{Loeliger_Averaging_Bound}. In \cite{Zamir_lattice_quantization}, several results regarding lattice quantization noise in high resolution has been derived and the problem of constructing lattices with an arbitrary quantization noise distribution has been studied in \cite{Gariby_Erez_Quantization_noise}.

Nested lattice codes were introduced in \cite{zamir_Nested_Lattice} where the concept of structured binning is presented. Nested linear/lattice code are important because in many communication problems, specially multi-terminal settings, such codes can be superior in average performance compared to random codes \cite{dinesh_dsc}. It has been shown in \cite{zamir_nested_lattice_wyner_ziv} that nested lattice codes are optimal for the Wyner-Ziv problem when the source and side information are jointly Gaussian. The dual problem of channel coding with state information has been addressed in \cite{urbanke_lattice_Gaussian_Channel} and the optimality of lattice codes for Gaussian channels has been shown. In \cite{Arun_nested_linear} it has been shown that nested linear codes are optimal for discrete channels with state information at the transmitter.

In this paper we focus on two problems: 1) The point to point channel coding with state information at the encoder (the Gelfand-Pinsker problem \cite{Gelfand_Pinsker}) and 2) Lossy source coding with side information at the decoder (the Winer-Ziv problem \cite{Wyner_Ziv_Discrete} \cite{Wyner_Ziv_Continuous}). We consider these two problems in their most general settings i.e. when the source and the channel are arbitrary. We use nested lattice codes with \emph{joint typicality decoding} rather than lattice decoding. We show that in both settings, from an information-theoretic point of view, nested lattice codes are optimal.

The paper is organized as follows: in Section \ref{section:prelimiraries} we present the required preliminaries and introduce our notation. In Section \ref{section:channel_coding} we show the optimality of nested lattice codes for channels with state information (the Gelfand-Pinsker problem). We show the optimality of nested lattice codes for source coding with side information (the Wyner-Ziv problem) in Section \ref{section:source_coding} and we finally conclude in Section \ref{section:conclusion}.

\section{Preliminaries}\label{section:prelimiraries}
\subsubsection{Channel Model}
We consider continuous memoryless channels with knowledge of channel state information at the transmitter used without feedback. We associate two sets $\mathcal{X}$ and $\mathcal{Y}$ with the channel as the channel input and output alphabets. The set of channel states is denoted by $\mathcal{S}$ and it is assumed that the channel state is distributed over $\mathcal{S}$ according to $P_S$. When the state of the channel $S$ is $s\in\mathcal{S}$, the input-output relation of the channel is characterized by a transition kernel  $W_{Y|XS}(y|x,s)$ for $x\in \mathcal{X}$ and $y\in \mathcal{Y}$. We assume that the state of the channel is known at the transmitter non-causally. The channel is specified by $(\mathcal{X},\mathcal{Y},\mathcal{S},P_S,W_{Y|XS},w)$ where $w:\mathcal{X}\times \mathcal{S}\rightarrow \mathds{R}^+$ is the cost function.\\

\subsubsection{Source Model}
The source is modeled as a discrete-time random process $X$ with each sample
taking values in a fixed set $\mathcal{X}$ called alphabet. Assume $X$ is distributed jointly with a random variable $S$ according to the measure $P_{XS}$ over $\mathcal{X}\times \mathcal{S}$ where $\mathcal{S}$ is an arbitrary set. We assume that the side information $S$ is known to the receiver non-causally. The reconstruction alphabet is denoted by $\mathcal{U}$ and the quality of reconstruction is measured by a single-letter distortion functions $d:\mathcal{X}\times \mathcal{U}\rightarrow \mathds{R}^{+}$. We denote such sources by $(\mathcal{X},\mathcal{S},\mathcal{U},P_{XS},d)$.\\

\subsubsection{Linear and Coset Codes Over $\mathds{Z}_p$} For a prime number $p$, a linear code over $\mathds{Z}_p$ of length $n$ and rate $R=\frac{k}{n}\log p$ is a collection of $p^k$ codewords of length $n$ which is closed under mod-$p$ addition (and hence mod-$p$ multiplication). In other words, linear codes over $\mathds{Z}_p$ are subspaces of $\mathds{Z}_p^n$. Any such code can be characterized by its generator matrix $G\in \mathds{Z}_p^{k\times n}$. This follow from the fact that any subgroup of an Abelian group corresponds to the image of a homomorphism into that group. The linear encoder maps a message tuple $u\in \mathds{Z}_p^k$ to the codeword $x$ where $x=uG$ and the operations are done mod-$p$. The set of all message tuples for this code is $\mathds{Z}_p^k$ and the set of all codewords is the range of the matrix $G$. i.e.
\begin{align}\label{eqn:linear_code_generator}
\mathds{C}=\left\{uG|u \in\mathds{Z}_p^k\right\}
\end{align}
A coset code over $\mathds{Z}_p$ is a shift of a linear code by a fixed vector. A coset code of length $n$ and rate $R=\frac{k}{n}\log p$ is characterized by its generator matrix $G\in \mathds{Z}_p^{k\times n}$ and it's shift vector (dither) $B\in \mathds{Z}_p^n$. The encoding rule for the corresponding coset code is given by $x=uG+B$, where $u$ is the message tuple and $x$ is the corresponding codeword. i.e.
\begin{align}\label{eqn:coset_code_generator}
\mathds{C}=\left\{uG+B|u \in\mathds{Z}_p^k\right\}
\end{align}

In a similar manner, any linear code over $\mathds{Z}_p$ of length $n$ and rate (at least) $R=\frac{n-k}{n}\log p$ is characterized by its parity check matrix $H\in \mathds{Z}_p^{k\times n}$. This follows from the fact that any subgroup of an Abelian group corresponds to the kernel of a homomorphism from that group. The set of all codewords of the code is the kernel of the matrix $H$; i.e.
\begin{align}\label{eqn:linear_code_parity}
\mathds{C}=\left\{u\in\mathds{Z}_p^n|Hu=0\right\}
\end{align}
where the operations are done mod-$p$. Note that there are at least $p^{n-k}$ codewords in this set. A coset code over $\mathds{Z}_p$ is a shift of a linear code by a fixed vector. A coset code of length $n$ and rate (at least) $R=\frac{n-k}{n}\log p$ can be characterized by its parity check matrix $H\in \mathds{Z}_p^{k\times n}$ and it's bias vector $c\in \mathds{Z}_p^k$ as follows:
\begin{align}\label{eqn:coset_code_parity}
\mathds{C}=\left\{u\in\mathds{Z}_p^n|Hu=c\right\}
\end{align}
where the operations are done mod-$p$.\\

\subsubsection{Lattice Codes and Shifted Lattice Codes} A lattice code of length $n$ is a collection of codewords in $\mathds{R}^n$ which is closed under real addition. A shifted lattice code is any translation of a lattice code by a real vector. In this paper, we use coset codes to construct (shifted) lattice codes as follows: Given a coset code $\mathds{C}$ of length $n$ over $\mathds{Z}_p$ and a \emph{step size} $\gamma$, define
\begin{align}
\Lambda(\mathds{C},\gamma,p)=\gamma (\mathds{C}-\frac{p-1}{2})
\end{align}
Then the corresponding mod-$p$ lattice code $\bar{\Lambda}(\mathds{C},\gamma,p)$ is the disjoint union of shifts of $\Lambda$ by vectors in $\gamma p\mathds{Z}^n$. i.e.
\begin{align*}
\bar{\Lambda}(\mathds{C},\gamma,p)=\displaystyle{\bigcup_{v\in p\mathds{Z}^n}}(\gamma v+\Lambda)
\end{align*}
It can be shown that this definition is equivalent to:
\begin{align*}
\bar{\Lambda}(\mathds{C},\gamma,p)=\left\{\gamma (v-\frac{p-1}{2})\left| v\in\mathds{Z}^n, v\mod p\in \mathds{C}\right.\right\}
\end{align*}
Note that $\Lambda(\mathds{C},\gamma,p)\subseteq \bar{\Lambda}(\mathds{C},\gamma,p)$ is a scaled and shifted copy of the linear code $\mathds{C}$.\\

\subsubsection{Nested Linear Codes}
A nested linear code consists of two linear codes, with the property than one of the codes (the \emph{inner linear code}) is a subset of the other code (the \emph{outer linear code}). For positive integers $k$ and $l$, let the outer and inner codes $\mathds{C}_i$ and $\mathds{C}_o$ be linear codes over $\mathds{Z}_p$ characterized by their generator matrices $G\in\mathds{Z}_p^{l\times n}$ and $G^\prime\in\mathds{Z}_p^{(k+l)\times n}$ and their shift vectors $B\in \mathds{Z}_p^n$ and $B^\prime\in \mathds{Z}_p^n$ respectively. Furthermore, assume 
\begin{align*}
G^\prime=\left[\!\!\begin{array}{c}G\\\Delta G\end{array}\!\!\right],\quad B^\prime=B
\end{align*}
For some $\Delta G\in \mathds{Z}_p^{k\times n}$.
In this case,
\begin{align}
&\label{eqn:Linear_codes1}\mathds{C}_o=\left\{aG+m\Delta G+B|a\in \mathds{Z}_p^l,m\in\mathds{Z}_p^k\right\},\\
&\label{eqn:Linear_codes2}\mathds{C}_i=\left\{aG+B|a\in \mathds{Z}_p^l\right\}
\end{align}
It is clear that the inner code is contained in the outer code. Furthermore, the inner code induces a partition of the outer code through its shifts. For $m\in \mathds{Z}_p^k$ define the $m$th \emph{bin} of $\mathds{C}_i$ in $\mathds{C}_o$ as 
\begin{align*}
&\mathds{B}_m=\left\{aG+m\Delta G+B|a\in \mathds{Z}_p^l\right\}
\end{align*}

Similarly, Nested linear codes can be characterized by the parity check representation of linear codes. For positive integers $k$ and $l$, let the outer and inner codes $\mathds{C}_o$ and $\mathds{C}_i$ be linear codes over $\mathds{Z}_p$ characterized by their parity check matrices $H\in\mathds{Z}_p^{l\times n}$ and $H^\prime\in\mathds{Z}_p^{(k+l)\times n}$ and their bias vectors $c\in \mathds{Z}_p^l$ and $c^\prime\in \mathds{Z}_p^{k+l}$ respectively. Furthermore assume:
\begin{align*}
H^\prime=\left[\begin{array}{c}H\\ \Delta H\end{array}\right], c^\prime=\left[\begin{array}{c}c\\ \Delta c\end{array}\right]
\end{align*}
For some $\Delta H\in \mathds{Z}_p^{k\times n}$ and $\Delta c\in \mathds{Z}_p^{k}$. In this case,
\begin{align}
&\label{eqn:Linear_codes3}\mathds{C}_o=\left\{u\in\mathds{Z}_p^n|Hu=c\right\},\\
&\label{eqn:Linear_codes4}\mathds{C}_i=\left\{u\in\mathds{Z}_p^n|Hu=c,\Delta H u=\Delta c\right\}
\end{align}
For $m\in \mathds{Z}_p^k$ define the $m$th bin of $\mathds{C}_i$ in $\mathds{C}_o$ as
\begin{align*}
&\mathds{B}_m=\left\{u\in\mathds{Z}_p^n|Hu=c,\Delta H u=m\right\}
\end{align*}
The outer code is the disjoint union of all the bins and each bin index $m\in \mathds{Z}_p^k$ is considered as a message. We denote a nested linear code by a pair $(\mathds{C}_i,\mathds{C}_o)$.\\
\subsubsection{Nested Lattice Codes} Given a nested linear code $(\mathds{C}_i,\mathds{C}_o)$ over $\mathds{Z}_p$ and a step size $\gamma$, define
\begin{align}
&\label{eqn:Lattice:1copy1}\Lambda_i(\mathds{C}_i,\gamma,p)=\gamma(\mathds{C}_i-\frac{p-1}{2}),\\
&\label{eqn:Lattice:1copy2}\Lambda_o(\mathds{C}_o,\gamma,p)=\gamma(\mathds{C}_o-\frac{p-1}{2})
\end{align}
Then the corresponding nested lattice code consists of an inner lattice code and an outer lattice code
\begin{align}\label{eqn:InnerLattice}
\bar{\Lambda}_i(\mathds{C}_i,\gamma,p)=\cup_{v\in p\mathds{Z}^n}(\gamma v+\Lambda_i)\\
%
\bar{\Lambda}_o(\mathds{C}_o,\gamma,p)=\cup_{v\in p\mathds{Z}^n}(\gamma v+\Lambda_o)
\end{align}
In this case as well, the inner lattice code induces a partition of the outer lattice code. For $m\in \mathds{Z}_p^k$, define
\begin{align}
\mathfrak{B}_m=\gamma (\mathds{B}_m-\frac{p-1}{2})
\end{align}
where $\mathds{B}_m$ is the $m$th bin of $\mathds{C}_i$ in $\mathds{C}_o$. The $m$th bin of the inner lattice code in the outer lattice code is defined by:
\begin{align*}
\bar{\mathfrak{B}}_m=\cup_{v\in p\mathds{Z}^n}(\gamma v+\mathfrak{B}_m)
\end{align*}
The set of messages consists of the set of all bins of $\bar{\Lambda}_i$ in $\bar{\Lambda}_o$. We denote a nested lattice code by a pair $(\bar{\Lambda}_i,\bar{\Lambda}_o)$.\\

\subsubsection{Achievability for Channel Coding and the Capacity-Cost Function}
A transmission system with parameters $(n,M,\Gamma,\tau)$ for reliable communication over a given channel $(\mathcal{X},\mathcal{Y},\mathcal{S},P_S,W_{Y|XS},w)$ with cost function $w:\mathcal{X}\times \mathcal{S}\rightarrow \mathds{R}^+$ consists of an encoding mapping and a decoding mapping
\begin{align*}
&e:\mathcal{S}^n\times \{1,2,\ldots,M\}\rightarrow \mathcal{X}^n\\
&f:\mathcal{Y}^n\rightarrow\{1,2,\ldots,M\}
\end{align*}
such that for all $m=1,2,\ldots,M$, if $s=(s_1,\cdots,s_n)$ and $x=e(s,m)=(x_1,\cdots,x_n)$, then 
\begin{align*}
\frac{1}{n}\displaystyle\sum_{i=1}^n w(x_i,s_i)<\Gamma
\end{align*}
and
\begin{align*}
\mathds{E}_{P_S}\left\{\sum_{m=1}^{M}\frac{1}{M}Pr\left(f(Y^n)\ne m|X^n=e(S^n,m)\right)\right\}\le \tau
\end{align*}
Given a channel $(\mathcal{X},\mathcal{Y},\mathcal{S},P_S,W_{Y|XS},w)$, a pair of non negative numbers $(R,W)$ is said to be achievable if for all $\epsilon>0$ and for all sufficiently large $n$, there exists a transmission system for reliable communication with parameters $(n,M,\Gamma,\tau)$ such that
\begin{align*}
\frac{1}{n}\log M \ge R-\epsilon,\qquad\Gamma\le W+\epsilon,\qquad\tau\le \epsilon
\end{align*}
The optimal capacity cost function $C(W)$ is given by the supremum of C such that $(C,W)$ is achievable.\\

\subsubsection{Achievability for Source Coding and the Rate-Distortion Function}
A transmission system with parameters $(n,\Theta,\Delta,\tau)$ for compressing a given source $(\mathcal{X},\mathcal{S},\mathcal{U},P_{XS},d)$ consists of an encoding mapping and a decoding mapping 
\begin{align*}
&e:\mathcal{X}^n\rightarrow \{1,2,\cdots,\Theta\},\\
&g:\mathcal{S}^n\times \{1,2,\cdots,\Theta\}\rightarrow \mathcal{U}^n
\end{align*}
such that the following condition is met:
\begin{align*}
P\left(d(X^n,g(e(X^n)))>\Delta\right)\le \tau
\end{align*}
where $X^n$ is the random vector of length $n$ generated by the source. In this transmission system, $n$ denotes the block length, $\log \Theta$ denotes the number of channel uses, $\Delta$ denotes the distortion level and $\tau$ denotes the probability of exceeding the distortion level $\Delta$.\\
Given a source, a pair of non-negative real numbers $(R,D)$ is said to be achievable if there exists for every $\epsilon>0$, and for all sufficiently large numbers $n$ a transmission system with parameters $(n,\Theta,\Delta,\tau)$ for compressing the source such that
\begin{align*}
\frac{1}{n}\log \Theta\le R+\epsilon, \qquad \Delta\le D+\epsilon,\qquad \tau\le \epsilon
\end{align*}
The optimal rate distortion function $R^*(D)$ of the source is given by the infimum of the rates $R$ such that $(R,D)$ is achievable.\\

\subsubsection{Typicality}
We use the notion of weak* typicality with Prokhorov metric introduced in \cite{mitran_polish}. Let $M(\mathds{R}^d)$ be the set of probability measures on $\mathds{R}^d$. For a subset $A$ of $\mathds{R}^d$ define its $\epsilon$-neighborhood by 
\begin{align*}
A^{\epsilon}=\{x\in\mathds{R}^d|\exists y\in A \mbox{ such that }\|x-y\|<\epsilon\}
\end{align*}
where $\|\cdot\|$ denotes the Euclidean norm in $\mathds{R}^d$. The Prokhorov distance between two probability measures $P_1,P_2\in M(\mathds{R}^d)$ is defined as follows:
\begin{align*}
\pi_d(P_1,P_2)=&\inf \{\epsilon>0|P_1(A)<P_2(A^{\epsilon})+\epsilon\mbox{ and }\\
&P_2(A)<P_1(A^{\epsilon})+\epsilon \quad \forall\mbox{ Borel set $A$ in }\mathds{R}^d\}
\end{align*}
Consider two random variables $X$ and $Y$ with joint distribution $P_{XY}(\cdot,\cdot)$ over $\mathcal{X}\times\mathcal{Y}\subseteq \mathds{R}^2$. Let $n$ be an integer and $\epsilon$ be a positive real number. For the sequence pair $(x,y)$ belonging to $\mathcal{X}^n\times \mathcal{Y}^n$ where $x=(x_1,\cdots,x_n)$ and $y=(y_1,\cdots,y_n)$ define the empirical joint distribution by
\begin{align*}
\bar{P}_{xy}(A,B)=\frac{1}{n}\sum_{i=1}^{n}\mathds{1}_{\{x_i\in A, y_i\in B\}}
\end{align*}
for Borel sets $A$ and $B$. Let $\bar{P}_{x}$ and $\bar{P}_y$ be the corresponding marginal probability measures. It is said that the sequence $x$ is weakly* $\epsilon$-typical with respect to $P_X$ if 
\begin{align*}
\pi_1(\bar{P}_x ,P_X)<\epsilon
\end{align*}
We denote the set of all weakly* $\epsilon$-typical sequences of length $n$ by $A_{\epsilon}^n(X)$. Similarly, $x$ and $y$ are said to be jointly weakly* $\epsilon$-typical with respect to $P_{XY}$ if 
\begin{align*}
\pi_2(\bar{P}_{xy},P_{XY})<\epsilon
\end{align*}
We denote the set of all weakly* $\epsilon$-typical sequence pairs of length $n$ by $A_{\epsilon}^n(XY)$.\\
Given a sequence $x\in A_{\epsilon}^n$, the set of conditionally $\epsilon$-typical sequences $A_\epsilon^n(Y|x)$ is defined as
\begin{align*}
A_\epsilon^n(Y|x)=\left\{y\in \mathcal{Y}^n\left| (x,y)\in
A_\epsilon^n(X,Y)\right.\right\}
\end{align*}

\subsubsection{Notation} In our notation, $O(\epsilon)$ is any function of $\epsilon$ such that $\lim_{\epsilon\rightarrow 0}O(\epsilon)=0$ and for a set $G$, $|G|$ denotes the cardinality (size) of $G$.\\

\section{Channel Coding}\label{section:channel_coding}
We show the achievability of the rate $R=I(U;Y)-I(U;S)$ for the Gelfand-Pinsker channel using nested lattice code for $U$.
\begin{theorem}\label{theorem:channel}
For the channel $(\mathcal{X},\mathcal{Y},\mathcal{S},P_S,W_{Y|XS},w)$, let $w:\mathcal{X}\rightarrow \mathds{R}^+$ be a continuous cost function. Let $\mathcal{U}$ be an arbitrary set and let $SUXY$ be distributed over $\mathcal{S}\times\mathcal{U}\times\mathcal{X}\times\mathcal{Y}$ according to $P_SP_{U|S}W_{X|US}W_{Y|SX}$ where the conditional distribution $P_{U|S}$ and the transition kernel $W_{X|US}$ are such that $\mathds{E}\{w(X)\}\le W$. Then the pair $(R,W)$ is achievable using nested lattice codes over $U$ where $R=I(U;Y)-I(U;S)$.
\end{theorem}
\subsection{Discrete $U$ and Bounded Continuous Cost Function}\label{section:discrete_U}
In this section we prove the theorem for the case when $U=\hat{U}$ takes values from the discrete set $\gamma(\mathds{Z}_p-\frac{p-1}{2})$ where $p$ is a prime and $\gamma$ is a positive number. We use a random coding argument over the ensemble of mod-$p$ lattice codes to prove the achievability. Let $\mathds{C}_o$ and $\mathds{C}_i$ be defined as (\ref{eqn:Linear_codes1}) and (\ref{eqn:Linear_codes2}) where $G$ is a random matrix in $\mathds{Z}_p^{l\times n}$, $\Delta G$ is a random matrix in $\mathds{Z}_p^{k\times n}$ and $B$ is a random vector in $\mathds{Z}_p^{n}$. Define $\bar{\Lambda}_i(\mathds{C}_i,\gamma,p)$ and $\bar{\Lambda}_o(\mathds{C}_o,\gamma,p)$ accordingly. The ensemble of nested lattice codes consists of all lattices of the form (\ref{eqn:Lattice:1copy1}) and (\ref{eqn:Lattice:1copy2}). The set of messages consists of all bins $\mathfrak{B}_m$ indexed by $m\in \mathds{Z}_p^k$.\\
The encoder observes the massage $m\in \mathds{Z}_p^k$ and the channel state $s\in \mathcal{S}^n$ and looks for a vector $u$ in the $m$th bin $\mathfrak{B}_m$ which is jointly weakly* typical with $s$ and encodes the massage $m$ to $x$ according to $W_{X|SU}$. The encoder declares error if it does not find such a vector.\\
After receiving $y\in\mathcal{Y}^n$, the decoder decodes it to $m\in\mathds{Z}_p^k$ if $m$ is the unique tuple such that the $m$th bin $\mathfrak{B}_m$ contains a sequence jointly typical with $y$. Otherwise it declares error.\\

\subsubsection{Encoding Error}
We begin with some definitions and lemmas. Let
\begin{align}\label{eqn:S_prime}
S^\prime=[\frac{-\gamma p}{2},\frac{\gamma p}{2}]^n\cap \gamma\mathds{Z}^n
\end{align}
For $a\in\mathds{Z}_p^k$, $m\in\mathds{Z}_p^l$, define
\begin{align*}
g(a,m)=\gamma\left((aG+m\Delta G+B)-\frac{(p-1)}{2}\right)
\end{align*}
$g(a,m)$ has the following properties:

\begin{lemma}
For $a\in\mathds{Z}_p^l$ and $m\in\mathds{Z}_p^k$, $g(a,m)$ is uniformly distributed over $S^\prime$. i.e. For $u\in S^\prime$,
\begin{align*}
P(g(a,m)=u)=\frac{1}{p^n}
\end{align*}
\end{lemma}
\begin{proof}
Note that $B$ is independent of $G$ and $\Delta G$ and therefore $aG+m\Delta G+B$ is a uniform variable over $\mathds{Z}_p^n$. The lemma follows by noting that
\begin{align*}
S^\prime=\gamma\left(\mathds{Z}_p^n-\frac{(p-1)}{2}\right)
\end{align*}
\end{proof}

\begin{lemma}
For $a,\tilde{a}\in\mathds{Z}_p^l$ and $m\in\mathds{Z}_p^k$ if $a\ne \tilde{a}$ then $g(a,m)$ and $g(\tilde{a},m)$ are independent. i.e. For $u\in S^\prime$ and $\tilde{u}\in S^\prime$,
\begin{align*}
P(g(a,m)=u,g(\tilde{a},m)=\tilde{u})=\frac{1}{p^{2n}}
\end{align*}
\end{lemma}
\begin{proof}
It suffices to show that $aG+m\Delta G+B$ and $\tilde{a}G+m\Delta G+B$ are uniform over $\mathds{Z}_p^n$ and independent. Note that for $u,\tilde{u}\in\mathds{Z}_p^n$,
\begin{align*}
&P\left(aG+m\Delta G+B=u,\tilde{a}G+m\Delta G+B=\tilde{u}\right)\\
&\qquad =P\left(aG+m\Delta G+B=u,(\tilde{a}-a)G=\tilde{u}-u\right)\\
&\qquad \stackrel{(a)}{=}P\left(aG+m\Delta G+B=u\right)\times P\left((\tilde{a}-a)G=\tilde{u}-u\right)\\
&\qquad \stackrel{(b)}{=}\frac{1}{p^{2n}}
\end{align*}
where $(a)$ follows since the $B$ is uniform over $\mathds{Z}_p^n$ and independent of $G$ and $(b)$ follows since $B$ and $G$ are uniform and $\tilde{a}-a\ne 0$
\end{proof}

\begin{lemma}\label{lemma:m_tildem_indep}
For $a,\tilde{a}\in\mathds{Z}_p^l$ and $m,\tilde{m}\in\mathds{Z}_p^k$ if $m\ne \tilde{m}$ then $g(a,m)$ and $g(\tilde{a},\tilde{m})$ are independent. i.e. For $u\in S^\prime$ and $\tilde{u}\in S^\prime$,
\begin{align*}
P(g(a,m)=u,g(\tilde{a},\tilde{m})=\tilde{u})=\frac{1}{p^{2n}}
\end{align*}
\end{lemma}
\begin{proof}
The proof is similar to the proof of the previous lemma and is omitted.
\end{proof}

For a message $m\in\mathds{Z}_p^k$ and state $s\in \mathcal{S}^n$, the encoder declares error if there is no sequence in $\mathfrak{B}_m$ jointly typical with $s$. Define
\begin{align*}
\theta(s)&=\sum_{u\in \mathfrak{B}_m}\mathds{1}_{\{u\in A_{\epsilon}^n(\hat{U}|s)\}}=\sum_{a\in \mathds{Z}_p^l}\mathds{1}_{\{g(a,m)\in A_{\epsilon}^n(\hat{U}|s)\}}
\end{align*}
Let $Z$ be a uniform random variable over $\gamma\left(\mathds{Z}_p-\frac{(p-1)}{2}\right)$ and hence $Z^n$ a uniform random variable over $S^\prime$. Then we have
\begin{align*}
\mathds{E}\{\theta(s)\}=\sum_{a\in \mathds{Z}_p^l}P{\left(Z^n\in A_{\epsilon}^n(\hat{U}|s)\right)}
\end{align*}
we need the following lemmas from to proceed:

\begin{lemma}\label{lemma:mitran1}
Let $P_{XY}$ be a joint distribution on $\mathds{R}^2$ and $P_X$ and $P_Y$ denote its marginals. Let $Z^n$ be a random sequence drawn according to $P_Z^n$. If $D(P_{XY}\|P_Z P_Y)$ is finite then for each $\delta>0$, there exist $\epsilon(\delta)$ such that if $\epsilon<\epsilon(\delta)$ and $y\in A_{{\epsilon}}^n(P_Y)$ then
\begin{align*}
\lim\sup \frac{1}{n}\log P_Z^n\!\left((Z^n,y)\!\in\! A_\epsilon^n(P_{XY}\right)\!\le \!-D(P_{XY}\|P_ZP_Y)\!+\!\delta
\end{align*}
\end{lemma}
\begin{proof}
This lemma is a generalization of Theorem 21 of \cite{mitran_polish}. The proof is provided in the Appendix.
\end{proof}

\begin{lemma}\label{lemma:mitran2}
Let $P_{XY}$ be a joint distribution on $\mathds{R}^2$ and $P_X$ and $P_Y$ denote its marginals. Let $Z^n$ be a random sequence drawn according to $P_Z^n$. Then for each $\epsilon,\delta>0$, there exist $\bar{\epsilon}(\epsilon,\delta)$ such that if $y\in A_{\bar{\epsilon}}^n(P_Y)$ then
\begin{align*}
\lim\inf \frac{1}{n}\log P_Z^n\!\left((Z^n,y)\!\in\!A_\epsilon^n(P_{XY}\right)\!\ge \!-D(P_{XY}\|P_ZP_Y)\!-\!\delta
\end{align*}
\end{lemma}
\begin{proof}
This lemma is a generalization of Theorem 22 of \cite{mitran_polish}. The proof is provided in the Appendix.
\end{proof}

Using these lemmas we get
\begin{align*}
\mathds{E}\{\theta(s)\}=p^l 2^{-n[D(P_{\hat{U}S}\|P_ZP_S)+O(\epsilon)]}
\end{align*}
Similarly, let $Z^n=g(a,m)$ and $\tilde{Z}^n=g(\tilde{a},m)$. Note that $Z^n$ and $\tilde{Z}^n$ are equal if $a=\tilde{a}$ and are independent if $a\ne\tilde{a}$. We have
\begin{align*}
\mathds{E}\{\theta(s)^2\}&=\sum_{a,\tilde{a}\in \mathds{Z}_p^l}P{\left(Z^n,\tilde{Z}^n\in A_{\epsilon}^n(\hat{U}|s)\right)}\\
&=\sum_{a\in \mathds{Z}_p^l}P{\left(Z^n\in A_{\epsilon}^n(\hat{U}|s)\right)}\\
&\qquad \qquad+\sum_{\substack{a,\tilde{a}\in \mathds{Z}_p^l\\a\ne\tilde{a}}}P{\left(Z^n\in A_{\epsilon}^n(\hat{U}|s)\right)}^2\\
&=p^l 2^{-n[D(P_{\hat{U}S}\|P_ZP_S)+O(\epsilon)]}\\
&\qquad \qquad+p^l (p^l -1)2^{-2n[D(P_{\hat{U}S}\|P_ZP_S)+O(\epsilon)]}
\end{align*}
Therefore
\begin{align*}
\mbox{var}\{\theta(s)\}&=\mathds{E}\{\theta(s)^2\}-\mathds{E}\{\theta(s)\}^2\\
&\le p^l 2^{-n[D(P_{\hat{U}S}\|P_ZP_S)+O(\epsilon)]}
\end{align*}
Hence,
\begin{align*}
P(\theta(s)=0)&\le P\left(\left|\theta(s)-\mathds{E}\{\theta(s)\}\right|\ge \mathds{E}\{\theta(s)\right)\\
&\stackrel{(a)}{\le}\frac{\mbox{var}\{\theta(s)\}}{\mathds{E}\{\theta(s)\}^2}\\
&\le p^l 2^{-n[D(P_{\hat{U}S}\|P_ZP_S]+O(\epsilon)]}
\end{align*}
Where $(a)$ follows from Chebyshev's inequality. This bound is valid for all $s\in \mathds{S}^n$. Therefore if
\begin{align}\label{eqn:EncErr_Ch}
\frac{l}{n}\log p> D(P_{\hat{U}S}\|P_ZP_S)
\end{align}
then the probability of encoding error goes to zero as the block length increases.\\

\subsubsection{Decoding Error}
The decoder declares error if there is no bin $\mathfrak{B}_m$ containing a sequence jointly typical with $y$ where $y$ is the received channel output or if there are multiple bins containing sequences jointly typical with $y$. Assume that the message $m$ has been encoded to $x$ according to $W_{X|SU}$ where $u=g(a,m)$ and the channel state is $s$. The channel output $y$ is jointly typical with $u$ with high probability. Given $m,s,a$ and $u$, the probability of decoding error is upper bounded by
\begin{align*}
P_{err}&\le\! \sum_{\substack{\tilde{m}\in \mathds{Z}_p^k\\\tilde{m}\ne m}}\!\sum_{\tilde{a}\in \mathds{Z}_p^l}\!\!P\left(g(\tilde{a},\tilde{m})\in A_\epsilon^n(\hat{U}|y)|g(a,m)\in
A_\epsilon^n(\hat{U}|y)\right)\\
&\stackrel{(a)}{=}p^lp^k2^{-n[D(P_{\hat{U}Y}\|P_ZP_Y)+O(\epsilon)]}
\end{align*}
Where in $(a)$ we use Lemmas \ref{lemma:m_tildem_indep}, \ref{lemma:mitran1} and \ref{lemma:mitran2}. Hence the probability of decoding error goes to zero if
\begin{align}\label{eqn:DecErr_Ch}
\frac{k+l}{n}\log p< D(P_{\hat{U}Y}\|P_ZP_Y)
\end{align}

\subsubsection{The Achievable Rate} Using (\ref{eqn:EncErr_Ch}) and (\ref{eqn:DecErr_Ch}), we conclude that if we choose $\frac{l}{n}\log p$ sufficiently close to $D(P_{\hat{U}S}\|P_ZP_S)$ and $\frac{k+l}{n}\log p$ sufficiently close to $D(P_{\hat{U}S}\|P_ZP_S)$ we can achieve the rate
\begin{align*}
R&=\frac{k}{n}\log p \approx D(P_{\hat{U}Y}\|P_ZP_Y)-D(P_{\hat{U}S}\|P_ZP_S)\\
&=I(\hat{U};Y)-I(\hat{U};S)
\end{align*}

\subsection{Arbitrary $U$ and Bounded Continuous Cost Function}
Let $Q=\{A_1,A_2,\cdots,A_r\}$ be a finite measurable partition of $\mathds{R}^d$. For random variables $U$ and $Y$ on $\mathds{R}^d$ with measure $P_{UY}$ define the quantized random variables $U_Q$ and $Y_Q$ on $Q$ with measure 
\begin{align*}
P_{U_QY_Q}(A_i,A_j)=P_{UY}(A_i,A_j)
\end{align*}
The Kullback-Leibler divergence between $U$ and $Y$ is defined as 
\begin{align*}
D(U\|Y)=\sup_{Q}D(U_Q\|Y_Q)
\end{align*}
where $D(U_Q\|Y_Q)$ is the discrete Kullback-Leibler divergence and the supremum is taken over all finite partitions $Q$ of $\mathds{R}^d$.
Similarly, the mutual information between $U$ and $Y$ is defined as 
\begin{align*}
I(U;Y)=\sup_{Q}I(U_Q;Y_Q)
\end{align*}
where $I(U_Q;Y_Q)$ is the discrete mutual information between the two random variables and the supremum is taken over all finite partitions $Q$ of $\mathds{R}^d$.\\
We have shown in Section \ref{section:discrete_U} that for discrete random variables the region given in Theorem \ref{theorem:channel} is achievable. In this part, we make a quantization argument to generalize this result to arbitrary auxiliary random variables. Let $S,U,X,Y$ be distributed according to $P_SP_{U|S}W_{X|US}W_{Y|X}$ where in this case $U$ is an arbitrary random variable. We start with the following theorem:
\begin{theorem}\label{theorem:mutual_inf_continuous}
Let $\mathcal{F}_1\subseteq\mathcal{F}_2\subseteq\cdots$ be an increasing sequence of $\sigma$-algebras on a measurable set $A$. Let $\mathcal{F}_{\infty}$ denote the $\sigma$-algebra generated by the union $\cup_{n=1}^{\infty}\mathcal{F}_n$. Let $P$ and $Q$ be probability measures on $A$. Then
\begin{align*}
D(P|_{\mathcal{F}_n}\|Q|_{\mathcal{F}_n})\rightarrow D(P|_{\mathcal{F}_{\infty}}\|Q|_{\mathcal{F}_{\infty}})\mbox{ as }n\rightarrow \infty
\end{align*}
where $P|_{\mathcal{F}}$ denotes the restriction of $P$ on $\mathcal{F}$.
\end{theorem}
\begin{IEEEproof}
Provided in \cite{mutual_info_continuous_Harremoes} and \cite{Barron_mutual_info} for example.
\end{IEEEproof}

For a prime $p>2$, a real positive number $\gamma$ and for $i=0\cdots,p-1$ define 
\begin{align*}
a_i=\frac{-\gamma(p-1)}{2}+\gamma i
\end{align*}
Define the quantization $Q_{\gamma,p}$ as $Q_{\gamma,p}=\{A_0,A_2,\cdots,A_{p-1}\}$ where 
\begin{align*}
&A_0=(-\infty,a_0]\\
&A_i=(a_{i-1},a_i],\mbox{ for }i=1,\cdots,p-2\\
&A_{p-1}=(a_{p-2},+\infty)
\end{align*}
Let the random variable $\hat{U}_{\gamma,p}$ take values from $\{a_0,\cdots,a_{p-1}\}$ according to joint measure
\begin{align}\label{eqn:joint_distribution}
P_{S\hat{U}XY}(\hat{U}=a_i,SXY\in B)=P_{SUXY}(U\in A_i,SXY\in B)
\end{align}
For all Borel sets $B\subseteq \mathds{R}^3$. For a fixed $\gamma$, let $p\le q$ be two primes. Then the $\sigma$-algebra induced by $Q_{\gamma,p}$ is included in the $\sigma$-algebra induced by $Q_{\gamma,q}$. Therefore, for a fixed $\gamma$, we can use the above theorem to get
\begin{align}\label{eqn:first_I_limit}
I(U|_{\mathcal{F}_{\gamma,p}};Y|_{\mathcal{F}_{\gamma,p}})\rightarrow I(U|_{\mathcal{F}_{\gamma,\infty}};Y|_{\mathcal{F}_{\gamma,\infty}})\mbox{ as }p\rightarrow \infty
\end{align}
where $U|_{\mathcal{F}_{\gamma,\infty}}$ is a random variable over $Q_{\gamma,\infty}=\{A_i|i\in\mathds{Z}\}$ where $A_i=\frac{\gamma}{2}+(\gamma i,\gamma(i+1)]$ with measure $P_{U|_{\mathcal{F}_{\gamma,\infty}}}(A_i)=P_U(A_i)$.\\
Let $\gamma_0=1$ and define $\gamma_n=\frac{1}{2^n}$. Note that if $m>n$ then $\mathcal{F}_{\gamma_n,\infty}$ is included in $\mathcal{F}_{\gamma_m,\infty}$. Also, since dyadic intervals generate the Borel Sigma field (\cite{morters_brownian_motion} for example), the restriction of $U$ to the sigma algebra generated by $\cup_{n=1}^{\infty}\mathcal{F}_{\gamma_n,\infty}$ is $U$ itself. We can use Theorem \ref{theorem:mutual_inf_continuous} to get
\begin{align}\label{eqn:second_I_limit}
I(U|_{\mathcal{F}_{\gamma_n,\infty}};Y|_{\mathcal{F}_{\gamma_n,\infty}})\rightarrow I(U;Y)\mbox{ as }n\rightarrow \infty
\end{align}
Combining (\ref{eqn:first_I_limit}) and (\ref{eqn:second_I_limit}) we conclude that for all $\epsilon>0$, there exist $\Gamma$ and $P$ such that if $\gamma\le \Gamma$ and $p\ge \Gamma$ then
\begin{align*}
\left|I(U|_{\mathcal{F}_{\gamma,p}};Y|_{\mathcal{F}_{\gamma,p}})-I(U;Y)\right|<\epsilon
\end{align*}
Since quantization reduces the mutual information ($X_Q\rightarrow X\rightarrow Y$), we have
\begin{align*}
I(U|_{\mathcal{F}_{\gamma,p}};Y|_{\mathcal{F}_{\gamma,p}})\le I(U|_{\mathcal{F}_{\gamma,p}};Y)\le I(U;Y)
\end{align*}
Therefore $\left|I(U|_{\mathcal{F}_{\gamma,p}};Y)-I(U;Y)\right|<\epsilon$. Also note that $I(U|_{\mathcal{F}_{\gamma,p}};Y)= I(\hat{U}_{\gamma,p};Y)$ since we define the joint measure to be the same. Therefore
\begin{align}\label{eqn:I_U_Y}
\left|I(\hat{U}_{\gamma,p};Y)-I(U;Y)\right|\le \epsilon
\end{align}
With a similar argument, for all $\epsilon>0$ there exist $\gamma$ and $p$ such that
\begin{align}\label{eqn:I_U_S}
\left|I(\hat{U}_{\gamma,p};S)-I(U;S)\right|\le \epsilon
\end{align}
if we take the maximum of the two $p$'s and the minimum of the two $\gamma$'s, we can say for all $\epsilon>0$ there exist $\gamma$ and $p$ such that both (\ref{eqn:I_U_Y}) and (\ref{eqn:I_U_S}) happen.\\
consider the sequence $P_{S\hat{U}_{\gamma_n,p}X}$ as $n,p\rightarrow \infty$. In the next lemma we show that under certain conditions this sequence converges in the weak* sense to $P_{SUX}$.
\begin{lemma}
Consider the sequence $P_{S\hat{U}_{\gamma_n,p}X}$ where $n\rightarrow\infty$ and $p$ is such that $\gamma_n p\rightarrow \infty$ as $n\rightarrow\infty$ (Take $p$ to be the smallest prime larger than $2^{2n}$ for example.). Then the sequence converges to $P_{SUX}$ in the weak* sense as $n\rightarrow\infty$.
\end{lemma}
\begin{proof}
It suffices to show that the three dimensional cumulative distribution function $F_{S\hat{U}_{\gamma_n,p}X}$ converges to $F_{SUX}$ point-wise in all points $(s,u,x)\in\mathds{R}^3$ where $F$ is continuous. Let $(s,u,x)$ be a point where $F$ is continuous and for an arbitrary $\epsilon>0$, let $\delta$ be such that
\begin{align*}
\left|F_{SUX}(s,u-\delta,x)-F_{SUX}(s,u,x)\right|<\epsilon\\
\left|F_{SUX}(s,u+\delta,x)-F_{SUX}(s,u,x)\right|<\epsilon\\
\end{align*}
Let $p$ be such that $\gamma_n=\frac{1}{2^n}<\delta$ and find $p$ accordingly. Then there exist points $a_i,a_j$ such that $a_i\in [u-\delta,u]$ and $a_j\in [u,u+\delta]$. We have
\begin{align*}
F_{SUX}(s,u-\delta,x)&\le F_{S\hat{U}_{\gamma_n,p}X}(s,a_i,x)\\
&\le F_{S\hat{U}_{\gamma_n,p}X}(s,u,x)\\
&\le F_{S\hat{U}_{\gamma_n,p}X}(s,a_j,x)\\
&\le F_{SUX}(s,u+\delta,x)
\end{align*}
Therefore $\left|F_{S\hat{U}_{\gamma_n,p}X}(s,u,x)-F_{SUX}(s,u,x)\right|\le \epsilon$. This shows the point-wise convergence of $F_{S\hat{U}_{\gamma_n,p}X}$.
\end{proof}
The above lemma implies $\mathds{E}_{P_{S\hat{U}_{\gamma_n,p}X}}\{w(X,S)\}$ converges to $\mathds{E}_{P_{SUX}}\{w(X,S)\}\le W$ since $w$ is assumed to be bounded continuous.\\
We have shown that for arbitrary $P_{U|S}$ and $W_{X|SU}$, one can find $P_{\hat{U}|S}$ and $W_{X|S\hat{U}}$  induced from (\ref{eqn:joint_distribution}) such that $\hat{U}$ is a discrete variable and
\begin{align*}
&I(\hat{U};Y)-I(\hat{U};S)\approx I(U;Y)-I(U;S)\\
&\mathds{E}_{P_{S\hat{U}X}}\{w(X,S)\}\approx \mathds{E}_{P_{SUX}}\{w(X,S)\}
\end{align*}
Hence, using the result of section \ref{section:discrete_U}, we have shown the achievability of the rate region given in Theorem \ref{theorem:channel} for arbitrary auxiliary random variables when the cost function is bounded and continuous.
\subsection{Arbitrary $U$ and Continuous Cost Function}
For a positive number $l$, define the clipped random variable $\hat{X}$ by $\hat{X}=\mbox{sign}(X)\min(l,|X|)$ and let $\hat{Y}$ be distributed according to $W_{\hat{Y}|\hat{X}}(\cdot,\hat{x})=W_{Y|X}(\cdot,\hat{x})$.
\begin{lemma}
As $l\rightarrow \infty$, $I(U;\hat{Y})\rightarrow I(U;Y)$.
\end{lemma}
\begin{proof}
Note that for Borel sets $B_1,B_2,B_3$ if $B_2\subseteq (-l,l)$ then
\begin{align*}
P_{U\hat{X}\hat{Y}}(B_1,B_2,B_3)=P_{UXY}(B_1,B_2,B_3)
\end{align*}
For any $\epsilon>0$, let $Q=\{A_1,\cdots,A_r\}$ be a quantization such that
\begin{align*}
\left|I(U_Q;Y_Q)-I(U;Y)\right|<\epsilon
\end{align*}
For an arbitrary $\delta>0$, assume $l$ is large enough such that $P_X((-l,l))>1-\delta$. Then
\begin{align*}
&P_{U_QY_Q}(A_i,A_j)=P_{UXY}(A_i,\mathds{R},A_j)\\
&=P_{UXY}(A_i,(-l,l),A_j)\!+\!P_{UXY}(A_i,(-\infty,-l]\!\cup\![l,\infty),A_j)\\
&\le P_{UXY}(A_i,(-l,l),A_j)+P_{UXY}(\mathds{R},(-\infty,-l]\cup[l,\infty),\mathds{R})\\
&= P_{U\hat{X}\hat{Y}}(A_i,(-l,l),A_j)+P_{X}((-\infty,-l]\cup[l,\infty))\\
&\le P_{U\hat{Y}}(A_i,A_j)+\delta\\
&= P_{U_Q\hat{Y}_Q}(A_i,A_j)+\delta
\end{align*}
Also,
\begin{align*}
P_{U_QY_Q}(A_i,A_j)&=P_{UXY}(A_i,\mathds{R},A_j)\qquad\qquad\qquad\qquad\qquad\\
&\ge P_{UXY}(A_i,(-l,l),A_j)\\
&= P_{U\hat{X}\hat{Y}}(A_i,(-l,l),A_j)\\
&\ge P_{U\hat{X}\hat{Y}}(A_i,\mathds{R},A_j)-\delta\\
&= P_{U\hat{Y}}(A_i,A_j)-\delta\\
&= P_{U_Q\hat{Y}_Q}(A_i,A_j)-\delta
\end{align*}

Since the choice of $\delta$ is arbitrary and since the discrete mutual information is continuous, we conclude that as $\epsilon,\delta\rightarrow 0$ (hence $l\rightarrow \infty$), $I(U;\hat{Y})\rightarrow I(U;Y)$.
\end{proof}
Since $\hat{X}$ is bounded and $w$ is assumed to be continuous, $w$ is also bounded. This completes the proof.
\section{Source Coding}\label{section:source_coding}
In this section, we show the achievability of the rate $R=I(U;X)-I(U;S)$ for the Wyner-Ziv problem using nested lattice codes for $U$.
\begin{theorem}\label{th:approximate_markov_chain}
For the source $(\mathcal{X},\mathcal{S},\mathcal{U},P_{XS},d)$ assume $d:\mathcal{X}\times\mathcal{U}\rightarrow \mathds{R}^+$ is continuous. Let $U$ be a random variable taking values from the set $\mathcal{U}$ jointly distributed with $X$ and $S$ according to $P_{XS}W_{U|X}$ where $W_{U|X}(\cdot|\cdot)$ is a transition kernel. Further assume that there exists a measurable function $f:\mathcal{S}\times\mathcal{U}\rightarrow \mathcal{\hat{X}}$ such that $\mathds{E}\{d(X,f(S,U))\}\le D$. Then the rate $R^*(D)=I(X;U)-I(S;U)$
is achievable using nested lattice codes.
\end{theorem}

\subsection{Discrete $U$ and Bounded Continuous Distortion Function}
In this section we prove the theorem for the case when $U$ takes values from the discrete set $\gamma(\mathds{Z}_p-\frac{p-1}{2})$ where $p$ is a prime and $\gamma$ is a positive number. The generalization to the case where $U$ is arbitrary and the distortion function is continuous is similar to the channel coding problem and is omitted. We use a random coding argument over the ensemble of mod-$p$ lattice codes to prove the achievability. The ensemble of codes used for source coding is based on the parity check matrix representation of linear and lattice codes. Define the inner and outer linear codes as in (\ref{eqn:Linear_codes3}) and (\ref{eqn:Linear_codes4}) where $H$ is a random matrix in $\mathds{Z}_p^{{l}\times n}$, $\Delta H$ is a random matrix in $\mathds{Z}_p^{k\times n}$, $c$ is a random vector in $\mathds{Z}_p^{l}$ and $\Delta c$ is a random vector in $\mathds{Z}_p^{k}$. Define $\bar{\Lambda}_i(\mathds{C}_i,\gamma,p)$ and $\bar{\Lambda}_o(\mathds{C}_o,\gamma,p)$ accordingly. The set of messages consists of all bins $\mathfrak{B}_m$ indexed by $m\in \mathds{Z}_p^k$.\\
For $m\in \mathds{Z}_p^k$, Let $\mathfrak{B}_m$ be the $m$th bin of $\Lambda_i$ in $\Lambda_o$. The encoder observes the source sequence $x\in \mathcal{X}^n$ and looks for a vector $u$ in the outer code $\Lambda_o$ which is typical with $x$ and encodes the sequence $x$ to the bin of $\Lambda_i$ in $\Lambda_o$ containing $u$. The encoder declares error if it does not find such a vector.\\
Having observed the index of the bin $m$ and the side information $s$, the decoder looks for a unique sequence $u$ in the $m$th bin which is jointly typical with $s$ and outputs $f(u,s)$. Otherwise it declares error.\\

\subsubsection{Encoding Error}
Define $S^\prime$ as in (\ref{eqn:S_prime}). For $u\in S^\prime$ define
\begin{align*}
g(u)=\frac{1}{\gamma}u+\frac{p-1}{2}
\end{align*}
$g(u)$ has the following properties:

\begin{lemma}\label{lemma:uniform_lattice}
For $u\in S^\prime$,
\begin{align*}
P(u\in \Lambda_o)=P(Hg(u)=c)=\frac{1}{p^l}
\end{align*}
i.e. All points of $S^\prime$ lie on the outer lattice equiprobably.
\end{lemma}
\begin{IEEEproof}
Follows from the fact that $c$ is independent of $H$ and is uniformly distributed over $\mathds{Z}_p^l$.
\end{IEEEproof}

\begin{lemma}\label{lemma:uniform_indep_lattice}
For $u\in S^\prime$ and $\tilde{u}\in S^\prime$, if $u\ne \tilde{u}$,
\begin{align*}
P\left(u\in \Lambda_o,\tilde{u}\in \Lambda_o\right)=P\left(Hg(u)=c,Hg(\tilde{u})=c\right)=\frac{1}{p^{2l}}
\end{align*}
i.e. All points of $S^\prime$ lie on the outer lattice independently.
\end{lemma}
\begin{IEEEproof}
Note that
\begin{align*}
&P\left(Hg(u)=c,Hg(\tilde{u})=c\right)\\
&\qquad=P\left(Hg(u)=c,H(g(\tilde{u})-g(u))=0\right)\\
&\qquad\stackrel{(a)}{=}P\left(Hg(u)=c\right)\times P\left(H(g(\tilde{u})-g(u))=0\right)\\
&\qquad\stackrel{(b)}{=}\frac{1}{p^{2l}}
\end{align*}
Where $(a)$ follows since $c$ is uniform and independent of $H$ and $(b)$ follows since $H$ and $c$ are uniform and $g(\tilde{u})-g(u)$ is nonzero.
\end{IEEEproof}

For a source sequence $x\in\mathcal{X}^n$, the encoder declare error if there is no sequence $u\in \Lambda_o$ jointly typical with $x$. Define
\begin{align*}
\theta(x)&=\sum_{u\in \Lambda_o}\mathds{1}_{\{u\in A_{\epsilon}^n(\hat{U}|x)\}}
\end{align*}
Let $Z$ be a uniform random variable over $\gamma(\mathds{Z}_p-\frac{p-1}{2}))$ and $Z^n$ a uniform random variable over $S^\prime$. We need the following lemmas to proceed:

\begin{lemma}\label{lemma:size_of_lattice}
With the above construction $|\Lambda_o|=p^{n-l}$ with high probability. Specifically,
\begin{align*}
P\left(\mbox{rank}(H)=l\right)&=\frac{(p^n-1)(p^n-p)(p^n-p^2)\cdots (p^n-p^{l-1})}{p^{nl}}\\
\ge 1-\frac{1}{p^{n-l}}
\end{align*}
and hence the probability that $|\Lambda_o|=p^{n-l}$ is close to one if $n$ is large. Furthermore, for $i=1,2,\cdots,l$,
\begin{align*}
P\left(\mbox{rank}(H)=i\right)\le {l\choose i}\frac{p^{i(l-i)}}{p^{n(l-i)}}
\end{align*}
\end{lemma}
\begin{proof}
The first part of the lemma follows since the total number of choices for $H$ is equal to $p^{nl}$ and the number of choices with independent rows is equal to $(p^n-1)(p^n-p)(p^n-p^2)\cdots (p^n-p^{l-1})$. Now we show the upper bounds. For a matrix $H$ to have a rank $i$, there should exist $i$ independent rows and the rest of the rows must be a linear combination of these rows (There are $p^i$ of such linear combinations). Hence the total number of such matrices is upper bounded by
\begin{align*}
{l\choose i}(p^n-1)(p^n-p)(p^n-p^2)\cdots (p^n-p^{i-1})(p^i)^{l-i}
\end{align*}
The lemma follows if we upper bound this quantity by
\begin{align*}
{l\choose i} p^{ni} p^{i(l-i)}
\end{align*}
\end{proof}

\begin{lemma}
With $\theta(x)$ and $Z^n$ defined as above, we have
\begin{align*}
&\mathds{E}\{\theta(x)\}\le p^{n-l}P\left(Z^n\in A_{\epsilon}^n(\hat{U}|s)\right)+\frac{2^l}{p^{n(l-1)}}\\
&\mathds{E}\{\theta(x)\}\ge (1-\frac{1}{p^{n-l}})p^{n-l}P\left(Z^n\in A_{\epsilon}^n(\hat{U}|s)\right)
\end{align*}
\end{lemma}
\begin{proof}
Write the random lattice $\Lambda_o$ as $\{u_1(\Lambda_o),u_2(\Lambda_o),\cdots,u_{r}(\Lambda_o)\}$ where $r$ is the cardinality of $\Lambda_o$ and $u_1(\Lambda_o),u_2(\Lambda_o),\cdots,u_{r}(\Lambda_o)$ are picked without replacement from $\Lambda_o$. It follow from Lemma \ref{lemma:uniform_lattice} that given $|\Lambda_o|=r=p^{n-l}$, $u_1(\Lambda_o), u_2(\Lambda_o),\cdots,u_{r}(\Lambda_o)$ are each uniformly distributed random variables over $S^\prime$. To see this note that for arbitrary $u\in S^\prime$, since $u_1(\Lambda_o),u_2(\Lambda_o),\cdots,u_{r}(\Lambda_o)$ are picked randomly from $\Lambda_o$,
\begin{align*}
P\left(u=u_1(\Lambda_o)\right)=P\left(u=u_2(\Lambda_o)\right)=\cdots=P\left(u=u_r(\Lambda_o)\right)
\end{align*}
Therefore
\begin{align*}
P(u\in \Lambda_o)&=\sum_{i=1}^r P\left(u=u_i(\Lambda_o)\right)\\
&=r P\left(u=u_1(\Lambda_o)\right)
=\frac{1}{p^l}
\end{align*}
Hence if $r=p^{n-l}$ then $u_1(\Lambda_o)$ is uniform over $S^\prime$. This argument is valid for all $i=1,\cdots,r$ and hence if $r=p^{n-l}$ then $u_i(\Lambda_o)$ is uniform over $S^\prime$.
Note that
\begin{align*}
\mathds{E}\{\theta(x)\}&=\mathds{E}\{\mathds{E}\{\theta(x)||\Lambda_o|=r\}\}
\end{align*}
The conditional expectation on the right hand side of this equation is upper bounded by $p^{n-l}$ and for $r=p^{n-l}$ it is equal to
\begin{align*}
\mathds{E}\{\theta(x)||\Lambda_o|=p^{n-l}\}&=\mathds{E}\{\sum_{u\in \Lambda_o}\mathds{1}_{\{u\in A_{\epsilon}^n(\hat{U}|x)\}}\}\\
&=\mathds{E}\{\sum_{i=1}^{p^{n-l}}\mathds{1}_{\{u_i(\Lambda_o)\in A_{\epsilon}^n(\hat{U}|x)\}}\}\\
&=\sum_{i=1}^{p^{n-l}} P\left(u_i(\Lambda_o)\in A_{\epsilon}^n(\hat{U}|x)\right)\\
&\stackrel{(a)}{=}\sum_{i=1}^{p^{n-l}} P\left(Z^n\in A_{\epsilon}^n(\hat{U}|x)\right)\\
&=p^{n-l} P\left(Z^n\in A_{\epsilon}^n(\hat{U}|x)\right)
\end{align*}
Where $(a)$ follows since $u_i(\Lambda_o)$ is uniformly distributed over $S^\prime$ for all $i=1,\cdots,r$. Next note that
\begin{align*}
\mathds{E}\{\theta(x)\}&=\sum_{r=0}^{p^n}P\left(|\Lambda_o|=r\right)\mathds{E}\{\theta(x)||\Lambda_o|=r\}\\
&\le P\left(|\Lambda_o|=p^{n-l}\right)r P\left(Z^n\in A_{\epsilon}^n(\hat{U}|x)\right)\\
&\qquad +\sum_{i=0}^{l-1}P\left(|\Lambda_o|=p^{n-i}\right) p^{n-i}\\
&\le p^{n-l}P\left(Z^n\in A_{\epsilon}^n(\hat{U}|s)\right)+\frac{2^l}{p^{n(l-1)}}
\end{align*}
Similarly,
\begin{align*}
\mathds{E}\{\theta(x)\}&=\sum_{r=0}^{p^n}P\left(|\Lambda_o|=r\right)\mathds{E}\{\theta(x)||\Lambda_o|=r\}\\
&\ge P\left(|\Lambda_o|=p^{n-l}\right)r P\left(Z^n\in A_{\epsilon}^n(\hat{U}|x)\right)\\
&\ge (1-\frac{1}{p^{n-l}})p^{n-l}P\left(Z^n\in A_{\epsilon}^n(\hat{U}|s)\right)
\end{align*}
\end{proof}

Therefore,
\begin{align*}
\mathds{E}\{\theta(s)\}=p^{n-l}2^{-n[D(P_{\hat{U}X}\|P_ZP_X)+O(\epsilon)]}
\end{align*}
Similarly,
\begin{align*}
\theta(x)^2&=\sum_{u,\tilde{u}\in \Lambda_o}\mathds{1}_{\{u,\tilde{u}\in A_{\epsilon}^n(\hat{U}|x)\}}\\
&=\sum_{u\in \Lambda_o}\mathds{1}_{\{u\in A_{\epsilon}^n(\hat{U}|x)\}}+\sum_{u\ne\tilde{u}\in \Lambda_o}\mathds{1}_{\{u,\tilde{u}\in A_{\epsilon}^n(\hat{U}|x)\}}\\
&\le\sum_{u\in \Lambda_o}\mathds{1}_{\{u\in A_{\epsilon}^n(\hat{U}|x)\}}+\sum_{u,\tilde{u}\in \Lambda_o}\mathds{1}_{\{u,\tilde{u}\in A_{\epsilon}^n(\hat{U}|x)\}}\\
\end{align*}
It can be shown that
\begin{align*}
\mathds{E}\{\theta(x)^2\}&=\mathds{E}\{|\Lambda_o|\}P{\left(Z^n\in A_{\epsilon}^n(\hat{U}|x)\right)}\\
&+\mathds{E}\{|\Lambda_o|\}^2P{\left(Z^n\in A_{\epsilon}^n(\hat{U}|x)\right)^2}\\
&\le p^{n-l}2^{-n[D(P_{\hat{U}X}\|P_ZP_X)+O(\epsilon)]}\\ &+p^{2(n-l)}2^{-2n[D(P_{\hat{U}X}\|P_ZP_X)+O(\epsilon)]}
\end{align*}
Hence
\begin{align*}
\mbox{var}\{\theta(x)\}\le p^k2^{-n[D(P_{\hat{U}X}\|P_ZP_X)+O(\epsilon)]}
\end{align*}
Hence,
\begin{align*}
P(\theta(s)=0)\le \frac{\mbox{var}\{\theta(x)\}}{\mathds{E}\{\theta(x)\}^2}\le p^{-(n-l)}2^{n[D(P_{\hat{U}X}\|P_ZP_X)+O(\epsilon)]}
\end{align*}
Therefore if
\begin{align}\label{eqn:enc_err_source}
\frac{l}{n}\log p<\log p-D(P_{\hat{U}X}\|P_ZP_X)
\end{align}
then the probability of encoding error goes to zero as the block length increases.
\subsubsection{Decoding Error}
After observing $m$ and the side information $s$, the decoder declares error if it does not find a sequence in the bin $\mathfrak{B}_m$ jointly typical with $s$ or if there are multiple of such sequences. We will show that the probability that a sequence $\tilde{u}\ne u$ is in the same bin as $u$ and is jointly typical with $s$ goes to zero as the block length increases if $\frac{k+l}{n}\log p> \log p-D(P_{\hat{U}S}\|P_ZP_S)$. The probability of decoding error is upper bounded by
\begin{align*}
P_{err}&\le \sum_{\tilde{u}\in \prime{S}^n}P\left(u\in\mathfrak{B}_m,u\in A_\epsilon^n(\hat{U}|s)\right)\\
&=\sum_{\tilde{u}\in \prime{S}^n} P\left(u\in\mathfrak{B}_m\right)P\left(Z^n\in A_\epsilon^n(\hat{U}|s)\right)\\
&=\frac{p^n}{p^{k+l}}2^{-n[D(P_{\hat{U}S}\|P_ZP_S)+O(\epsilon)]}
\end{align*}
Hence the probability of decoding error goes to zero if
\begin{align}\label{eqn:enc_err_source}
\frac{k+l}{n}\log p>\log p- D(P_{\hat{U}S}\|P_ZP_S)
\end{align}

\subsubsection{The Achievable Rate} Using (\ref{eqn:enc_err_source}) and (\ref{eqn:enc_err_source}), we conclude that if we choose $\frac{l}{n}\log p$ sufficiently close to $\log p-D(P_{\hat{U}X}\|P_ZP_X)$ and $\frac{k+l}{n}\log p$ sufficiently close to $\log p-D(P_{\hat{U}S}\|P_ZP_S)$ we can achieve the rate
\begin{align*}
R&=\frac{k}{n}\log p\\
&\approx D(P_{\hat{U}X}\|P_ZP_X)-D(P_{\hat{U}S}\|P_ZP_S)\\
&=I(X;\hat{U})-I(S;\hat{U})
\end{align*}

\section{Conclusion}\label{section:conclusion}
We have shown that nested lattice codes are optimal for the Gelfand-Pinsker problem as well as the Wyner-Ziv problem.

\section{Appendix}
\subsection{Proof of Lemma \ref{lemma:mitran1}}
The proof follows along the lines of the proof of Theorem 21 of \cite{mitran_polish}. Let $Q=\{A_1,A_2,\cdots,A_r\}$ be a finite partition of $\mathds{R}$. Let $Q_{XYZ}$, $Q_{XY}$, $Q_{XZ}$, $Q_{YZ}$, $Q_X$, $Q_Y$ and $Q_Z$ be measures induced by this partition, corresponding to $P_{XYZ}$, $P_{XY}$, $P_{XZ}$, $P_{YZ}$, $P_X$, $P_Y$ and $P_Z$ respectively. For the random sequence $Z^n=(Z_1,\cdots,Z_n)$ and the deterministic sequence $y=(y_1,\cdots,y_n)$ let $\bar{Q}_y$ be the deterministic empirical measure of $y$ and define the random empirical measures
\begin{align*}
&\bar{Q}_{Zy}(A_i,A_j)=\frac{1}{n}\sum_{i=1}^n \mathds{1}_{\{Z_i\in A_i, y_i\in A_j\}}\\
&\bar{Q}_{Z}(A_i)=\frac{1}{n}\sum_{i=1}^n \mathds{1}_{\{Z_i\in A_i\}}
\end{align*}
for $i,j=1,2,\cdots,r$. As a property of weakly* typical sequences, for a fixed $\epsilon_1>0$, there exists a sufficiently small $\epsilon>0$ such that for a sequence pair $(x,y)\in A_{\epsilon}^n(XY)$ and for all $i,j=1,2,\cdots,r$,
\begin{align*}
\left|\bar{Q}_{xy}(A_i,A_j)-Q_{XY}(A_i,A_j)\right|\le \epsilon_1
\end{align*}
where $\bar{Q}_{xy}$ is the joint empirical measure of $(x,y)$. It follows that the rare event $(Z^n,y)\in A_{\epsilon}^n(XY)$ is included in the intersection of events
\begin{align}\label{eqn:ijth_event}
\left\{\left|\bar{Q}_{Zy}(A_i,A_j)-Q_{XY}(A_i,A_j)\right|\le \epsilon_1\right\}
\end{align}
for $i,j=1,2,\cdots,r$. Therefore
\begin{align*}
&Q_Z^n\left((Z^n,y)\in A_{\epsilon}^n(XY)\right)\le\\
&\qquad\qquad Q_Z^n\left(\bigcap_{i,j=1}^r\left\{\left|\bar{Q}_{Zy}(A_i,A_j)-Q_{XY}(A_i,A_j) \right|\le \epsilon_1\right\}\right)
\end{align*}
Let $\epsilon(\delta)$ be such that for $j=1,\cdots,r$,
\begin{align*}
&\left|\bar{Q}_y(A_j)-Q_Y(A_j)\right|\le \epsilon_1\\
&1-\epsilon_1<\frac{\bar{Q}_y(A_j)}{Q_Y(A_j)}<1+\epsilon_1
\end{align*}
Note that if $Q_Y(A_j)=0$ then $Q_{XY}(A_i,A_j)=0$ and hence
\begin{align*}
\left|\bar{Q}_{Zy}(A_i,A_j)-Q_{XY}(A_i,A_j)\right|&=\bar{Q}_{Zy}(A_i,A_j)\\
&\le\bar{Q}_{y}(A_j)\le \epsilon_1
\end{align*}
and (\ref{eqn:ijth_event}) is satisfied. If we choose $\epsilon_1$ smaller than any nonzero $Q_Y(A_j)$ it follows that $\bar{Q}_y(A_j)>0$ whenever $Q_Y(A_j)>0$. Now assume that $Q_Y(A_j)>0$ and hence $\bar{Q}_y(A_j)>0$. Define
\begin{align*}
&Q_{X|Y}(A_i|A_j)=\frac{Q_{XY}(A_i,A_j)}{Q_Y(A_j)}\\
&\bar{Q}_{Z|y}(A_i|A_j)=\frac{\bar{Q}_{Zy}(A_i,A_j)}{\bar{Q}_y(A_j)}
\end{align*}
If $Q_Y(A_j)>0$, the event in (\ref{eqn:ijth_event}) is included in the event
\begin{align}\label{eqn:ijth_conditional}
&\nonumber\{\left|\bar{Q}_{Z|y}(A_i|A_j)\bar{Q}_y(A_j)-Q_{X|Y}(A_i|A_j)\bar{Q}_y(A_j)\right.\\
&\qquad\left.+Q_{X|Y}(A_i|A_j)\bar{Q}_y(A_j)-Q_{X|Y}(A_i|A_j)Q_Y(A_j)\right|\!\le\! \epsilon_1\!\}
\end{align}
Note that
\begin{align*}
&\left|Q_{X|Y}(A_i|A_j)\bar{Q}_y(A_j)-Q_{X|Y}(A_i|A_j)Q_Y(A_j)\right|\\
&\qquad \qquad\qquad \qquad = Q_{X|Y}(A_i|A_j)\left|\bar{Q}_y(A_j)-Q_Y(A_j) \right|\\
&\qquad \qquad\qquad \qquad\le \epsilon_1
\end{align*}
Therefore (\ref{eqn:ijth_conditional}) implies
\begin{align*}
\{\left|\bar{Q}_{Z|y}(A_i|A_j)\bar{Q}_y(A_j)\!-\!Q_{X|Y}(A_i|A_j)\right|\bar{Q}_y(A_j)\le 2\epsilon_1\}
\end{align*}
And this implies
\begin{align*}
\{\left|\bar{Q}_{Z|y}(A_i|A_j)\bar{Q}_y(A_j)-Q_{X|Y}(A_i|A_j)\right|\!\le\! \frac{2\epsilon_1}{\bar{Q}_y(A_j)(\!1\!-\!\epsilon_1\!)}\}
\end{align*}
Let
\begin{align*}
\epsilon_2=\max_{\substack{j=1\\Q_Y(A_j)>0}}^r \frac{2\epsilon_1}{\bar{Q}_y(A_j)(1-\epsilon_1)}
\end{align*}
then the event in (\ref{eqn:ijth_event}) is included in the event
\begin{align*}
\{\left|\bar{Q}_{Z|y}(A_i|A_j)\bar{Q}_y(A_j)-Q_{X|Y}(A_i|A_j)\right|\le \epsilon_2
\end{align*}
Therefore
\begin{align*}
&Q_Z^n\left((Z^n,y)\in A_{\epsilon}^n(XY)\right)\le\\
& Q_Z^n\left(\bigcap_{\substack{i,j=1\\Q_Y(A_j)>0}}^r\left\{\left|\bar{Q}_{Z|y}(A_i|A_j)-Q_{X|Y}(A_i|A_j) \right|\le \epsilon_2\right\}\right)
\end{align*}
Note that since $y$ is a deterministic sequence and $Z_i$'s are iid, the events
\begin{align*}
\left\{\left|\bar{Q}_{Z|y}(A_i|A_j)-Q_{X|Y}(A_i|A_j) \right|\le \epsilon_2\right\}
\end{align*}
are independent for different values of $j=1,\cdots,r$. Let $n_j=n \bar{Q}_y(A_j)$. Then,
\begin{align*}
&Q_Z^n\left((Z^n,y)\in A_{\epsilon}^n(XY)\right)\le\\
& \prod_{\substack{j=1\\Q_Y(A_j)>0}}^r Q_Z^{n_j}\left(\bigcap_{i=1}^r\left\{\left|\bar{Q}_{Z|y}(A_i|A_j)-Q_{X|Y}(A_i|A_j) \right|\le \epsilon_2\right\}\right)
\end{align*}
Since for $Q_Y(A_j)>0$, $n_j\rightarrow\infty$ as $n\rightarrow \infty$, it follows from Sanov's theorem \cite{Dembo_Large_Deviations} that
\begin{align*}
&\lim_{n\rightarrow \infty}\sup\frac{1}{n_j}\log\\
&\qquad\qquad Q_Z^{n_j}\left(\bigcap_{i=1}^r\left\{ \left|\bar{Q}_{Z|y}(A_i|A_j)-Q_{X|Y}(A_i|A_j) \right|\le \epsilon_2\right\}\right)\\
&\qquad\qquad\le -\left[D(Q_{X|Y}(\cdot|A_j)||Q_Z(\cdot))-\delta_{j}\right]
\end{align*}
where $\delta_j\rightarrow 0$ as $\epsilon_2\rightarrow 0$. Therefore
\begin{align*}
&\lim_{n\rightarrow \infty}\sup\frac{1}{n} \log Q_Z^n\left((Z^n,y)\in A_{\epsilon}^n(XY)\right)\\
&\le \sum_{\substack{j=1\\Q_Y(A_j)>0}}^r\lim_{n\rightarrow \infty}\sup\frac{n_j}{n}D(Q_{X|Y}(\cdot|A_j)||Q_Z(\cdot))\\
&\le \sum_{\substack{j=1\\Q_Y(A_j)>0}}^r \!\!\! -(1-\epsilon_1)Q_Y(A_j)\left[D(Q_{X|Y}(\cdot|A_j)||Q_Z(\cdot))\!-\!\delta_{j}\right]\\
&\le -(1-\epsilon_1)D(Q_{XY}||Q_ZQ_Y)+\delta^\prime
\end{align*}
where $\delta^\prime \rightarrow 0$ as $\epsilon_2\rightarrow 0$. For finite $D(P_{XY}||P_ZP_Y)$ the statement of the lemma follows by choosing the quantization $Q$ such that $D(Q_{XY}||Q_ZQ_Y)$ is sufficiently close to $D(P_{XY}||P_ZP_Y)$.

\subsection{Proof of Lemma \ref{lemma:mitran2}}
The proof follows along the lines of the proof of Theorem 22 of \cite{mitran_polish}. Let $Q=\{A_1,A_2,\cdots,A_r\}$ be a finite partition of $\mathds{R}$. Let $Q_{XYZ}$, $Q_{XY}$, $Q_{XZ}$, $Q_{YZ}$, $Q_X$, $Q_Y$ and $Q_Z$ be measures induced by this partition, corresponding to $P_{XYZ}$, $P_{XY}$, $P_{XZ}$, $P_{YZ}$, $P_X$, $P_Y$ and $P_Z$ respectively. For the random sequence $Z^n=(Z_1,\cdots,Z_n)$ and the deterministic sequence $y=(y_1,\cdots,y_n)$ let $\bar{Q}_y$ be the deterministic empirical measure of $y$ and define the random empirical measures
\begin{align*}
&\bar{Q}_{Zy}(A_i,A_j)=\frac{1}{n}\sum_{i=1}^n \mathds{1}_{\{Z_i\in A_i, y_i\in A_j\}}\\
&\bar{Q}_{Z}(A_i)=\frac{1}{n}\sum_{i=1}^n \mathds{1}_{\{Z_i\in A_i\}}
\end{align*}
For arbitrary $\delta>0$, let $Q$ be such that
\begin{align*}
&\pi(Q_{XY},P_{XY})< \epsilon\\
&\pi(Q_{ZY},P_{ZY})< \epsilon\\
&\left|D(P_{XY}||P_ZP_Y)-D(Q_{XY}||Q_ZQ_Y)\right|<\epsilon
\end{align*}
We show that for such a quantization, under certain conditions, the probability of the event
\begin{align*}
\left\{\pi(\bar{Q}_{Zy},Q_{XY})<\epsilon\right\}
\end{align*}
is close to the probability of the event
\begin{align*}
\left\{\pi(\bar{P}_{Zy},P_{XY})<5 \epsilon\right\}
\end{align*}
It follows from Theorem 18 of \cite{mitran_polish} that for arbitrary $\epsilon,\delta^\prime>0$, there exists some $\bar{\epsilon}>0$ such that for all $n$ greater than some $N$ if $y\in A_{\bar{\epsilon}}^n(Y)$, then
\begin{align*}
&\lim_{n\rightarrow\infty}P\left(\pi(\bar{P}_{Zy},P_{ZY})<\epsilon\right)>1-\delta\\
&\lim_{n\rightarrow\infty}P\left(\pi(\bar{Q}_{Zy},Q_{ZY})<\epsilon\right)>1-\delta
\end{align*}
Consider the event
\begin{align*}
\left\{\pi(\bar{Q}_{Zy},Q_{XY})<\epsilon, \pi(\bar{P}_{Zy},P_{ZY})<\epsilon, \pi(\bar{Q}_{Zy},Q_{ZY})<\epsilon\right\}
\end{align*}
This event implies
\begin{align*}
\pi(\bar{P}_{Zy},P_{XY})&\le \pi(\bar{P}_{Zy},P_{ZY})+ \pi(Q_{ZY},P_{ZY})\\
&+ \pi(\bar{Q}_{Zy},Q_{ZY})+ \pi(\bar{Q}_{Zy},Q_{XY})\\
&+ \pi(Q_{XY},P_{XY})\le 5\epsilon
\end{align*}
Therefore
\begin{align*}
&P\left(\pi(\bar{P}_{Zy},P_{XY})\le 5\epsilon\right)\ge\\
&P\left(\pi(\bar{Q}_{Zy},Q_{XY})<\epsilon, \pi(\bar{P}_{Zy},P_{ZY})<\epsilon, \bar{Q}_{Zy},Q_{ZY})<\epsilon\right)
\end{align*}
The right hand side can be lower bounded by
\begin{align}
&1-P\left(\pi(\bar{Q}_{Zy},Q_{XY})\ge \epsilon\right)\\
&- P\left(\pi(\bar{P}_{Zy},P_{ZY})\ge \epsilon\right)- P\left(\bar{Q}_{Zy},Q_{ZY})\ge \epsilon\right)\\
&\ge P\left(\pi(\bar{Q}_{Zy},Q_{XY})< \epsilon\right)-\delta-\delta
\end{align}
Note that for arbitrary $\delta^\prime$ and for sufficiently large $n$,
\begin{align*}
P\left(\pi(\bar{Q}_{Zy},Q_{XY})\right)\ge 2^{-n\left[D(Q_{XY}||Q_ZQ_Y)+\delta^\prime\right]}
\end{align*}
Since $\delta,\delta^\prime$ are arbitrary and $D(Q_{XY}||Q_ZQ_Y)\approx D(P_{XY}||P_ZP_Y)$, it follows that
\begin{align*}
P\left(\pi(\bar{P}_{Zy},P_{XY})\le 5\epsilon\right)\ge 2^{-n\left[D(P_{XY}||P_ZP_Y)+\delta+\epsilon^\prime\right]}-2\delta
\end{align*}

\bibliographystyle{IEEEtran}
\bibliography{IEEEabrv,ariabib}
\end{document}